\documentclass[fleqn,10pt]{wlscirep}
\usepackage{amsthm,amscd,float}
\usepackage[utf8]{inputenc}
\usepackage[T1]{fontenc}
\newcommand{\al}{\alpha}

\newcommand{\de}{\delta}

\newcommand{\vep}{\varepsilon}
\newcommand{\ga}{\gamma}

\newcommand{\si}{\sigma}
\renewcommand{\th}{\theta}

\newcommand{\bp}{\mathbf{p}}
\newcommand{\bq}{\mathbf{q}}

\newcommand{\bu}{\mathbf{u}}

\newcommand{\bv}{\mathbf{v}}

\newcommand{\bU}{\mathbf U}
\newcommand{\bV}{\mathbf V}

\newcommand{\ha}{\widehat a}
\newcommand{\had}{\ha^{\,\dagger}}
\newcommand{\hc}{\widehat c}
\newcommand{\hcd}{\hc^{\,\dagger}}
\newcommand{\hd}{\widehat d}
\newcommand{\hdd}{\hd^{\mkern 6mu\dagger}}

\newcommand{\hnu}{\widehat\nu}
\newcommand{\hC}{\widehat{C}}

\newcommand{\CC}{{\mathbb C}}

\newcommand{\cH}{{\mathcal H}}
\newcommand{\cI}{{\mathcal I}}

\newcommand{\cL}{{\mathcal L}}

\newcommand{\cO}{{\mathcal O}}

\newcommand{\id}{1\hspace{-.25em}{\rm l}}

\newcommand{\ket}[1]{|#1\rangle}
\newcommand{\bra}[1]{\langle#1|}
\newcommand{\kb}[2]{|#1\rangle\langle#2|}
\newcommand{\kbo}[1]{\kb{#1}{#1}}
\newcommand{\ksp}[3]{\bra{#1}#2\ket{#3}}
\newcommand{\Ksp}[3]{\big\langle #1\big|#2\big|#3\big\rangle}

\let\ds\displaystyle

\newcommand{\mss}{\kern 1pt}

\renewcommand{\le}{\leqslant}
\renewcommand{\ge}{\geqslant}
\newcommand{\tends}[1]{\bbuildrel{\hbox to 2em{\rightarrowfill}}_{#1}^{}}

\newcommand{\csch}{\operatorname{csch}}

\newcommand{\tr}{\operatorname{tr}}

\newcommand{\diag}{\operatorname{diag}}

\newcommand{\iu}{\mathrm i}
\newcommand{\e}{\mathrm e}
\newcommand{\diff}{\mathrm{d}}

\newcommand{\spec}{\operatorname{spec}}
\newcommand{\en}{\enspace}

\newcommand{\Int}[1]{\,\mathop{\!#1}\limits^{\lower1ex\hbox{$\scriptstyle\circ$}}{}}

\newcommand{\lra}{\leftrightarrow}
\newtheorem{thm}{Theorem}

\theoremstyle{remark}

\def\clap#1{\hbox to 0pt{\hss#1\hss}}

\newcommand{\tcup}{{\textstyle\bigcup}}
\newcommand{\finf}{f^{(\infty)}}

\title{A duality principle for the multi-block entanglement entropy of free fermion systems}

\author[1,+]{J.A. Carrasco}
\author[1,+]{F. Finkel}
\author[1,+,*]{A. Gonz\'alez-L\'opez}
\author[1,2,+]{P. Tempesta}
\affil[1]{Departamento de F\'\i sica Te\'orica II, Universidad Complutense de Madrid, 28040 Madrid, Spain}
\affil[2]{Instituto de Ciencias Matem\'aticas \textup(CSIC--UAM--UC3M--UCM\textup), c/ Nicol\'as
  Cabrera  13--15, 28049 Madrid, Spain}
\affil[*]{Corresponding author. Email: artemio@ucm.es}

\affil[+]{These authors contributed equally to this work}

\begin{abstract}
  The analysis of the entanglement entropy of a subsystem of a one-dimensional quantum system is a
  powerful tool for unravelling its critical nature. For instance, the scaling behaviour of the
  entanglement entropy determines the central charge of the associated Virasoro algebra. For a
  free fermion system, the entanglement entropy depends essentially on two sets, namely the
  set~$A$ of sites of the subsystem considered and the set~$K$ of excited momentum modes. In this
  work we make use of a general duality principle establishing the invariance of the entanglement
  entropy under exchange of the sets~$A$ and~$K$ to tackle complex problems by studying their dual
  counterparts. The duality principle is also a key ingredient in the formulation of a novel
  conjecture for the asymptotic behavior of the entanglement entropy of a free fermion system in
  the general case in which both sets~$A$ and~$K$ consist of an arbitrary number of blocks. We
  have verified that this conjecture reproduces the numerical results with excellent precision for
  all the configurations analyzed. We have also applied the conjecture to deduce several
  asymptotic formulas for the mutual and~$r$-partite information generalizing the known ones for
  the single block case.
\end{abstract}

\begin{document}

\flushbottom
\maketitle

\thispagestyle{empty}

\section*{Introduction}

One of the distinguishing features of the quantum realm is the existence of entangled states in
composite systems, which have no classical analogue and play a fundamental role in quantum
information theory and condensed matter physics~(see, e.g., Refs.~\cite{AFOV08,HHHH09}). A widely
used quantitative measure of the degree of entanglement between two subsystems~$A,B$ of a quantum
system $A\cup B$ in a pure state~$\rho=\ket\psi\bra\psi$ is the R\'enyi entanglement
entropy~\cite{Re61} $S_{\al}(A)=(1-\al)^{-1}\log\tr(\rho_A^\alpha)$, where~$\rho_A$ is the reduced
density matrix of the subsystem~$A$ and~$\al>0$ is the R\'enyi parameter (the von Neumann entropy is
obtained in the limit $\al\to1$). It is easy to show that~$S_\al(A)=S_\al(B)$, and that the
entanglement entropy vanishes when the whole system is in a non-entangled (product) state. Over
the last decade, it has become clear that the study of the entanglement between two extended
subsystems of a many-body system in one dimension is a powerful tool for uncovering its
criticality properties~\cite{VLRK03,JK04,IJK05,FIK08}. The reason for this is that one-dimensional
critical quantum systems are governed by an effective conformal field theory (CFT) in $(1+1)$
dimensions, whose entanglement entropy can be evaluated in closed form in the thermodynamic
limit~\cite{HLW94,CC04jstat,CC05}. In the simplest case, when the subsystem~$A$ consists of a
single interval of length~$L$ and the whole system is in its ground state, the scaling
of~$S_\al(A)$ for~$L\to\infty$ is determined solely by the central charge~$c$. In order to probe
the full operator content of the CFT, one needs to analyze more complicated situations in which
the set~$A$ is the union of a finite number of intervals. In fact, in the last few years there has
been a considerable interest in this problem, both for CFTs and one-dimensional lattice models
(integrable spin chains or free fermion systems), as witnessed by the number of papers published
on this subject~(see, e.g., Refs.~\cite{CFH05,CH09,CC09,CCT09,ATC10,FC10,CCT11,CTT14}).

In this work we shall extend this analysis to the more general case in which the system's state is
also made up of several blocks of consecutive excited momentum modes, which has received
comparatively less attention~\cite{EZ05,AFC09,KZ10,CIS11,AEFS14,CFGR17}. An important motivation
for dealing with this type of states is that it makes it possible to treat position and momentum
space on a more equal footing, thus revealing certain symmetries that have not been fully
exploited so far. This approach naturally leads to a duality principle for the behavior of
the entanglement entropy under the exchange of the position and momentum space block
configurations, which in fact can be exploited to solve problems that up until now had defied an
analytic treatment~\cite{AEF14} with standard techniques like the Fisher--Hartwig
conjecture~\cite{FH68}. We have applied this duality principle to propose a new conjecture on the
composability of the entanglement entropy in the multi-block case, which yields a closed
asymptotic formula for the R\'enyi entanglement entropy of a free fermion system in the most
general multi-block configuration, both in position and momentum space. This formula, which we
have numerically verified for a wide range of configurations both for~$0<\al<1$ and~$\al\ge1$,
reduces to the known ones when the configuration in momentum space consists of a single block. It
also leads to closed asymptotic formulas for the mutual and the tripartite~\cite{CH09} (or
$r$-partite~\cite{CTT14}) information, which again agree with the general CFT predictions.

\section*{Results and methods}

\subsection*{Preliminaries and notation}

The model considered is a system of $N$ free (spinless) hopping fermions with creation
operators~$a^\dagger_j$ (where the subindex~$j=0,\dots,N-1$ denotes the site) and Hamiltonian
$H=\sum_{i,j=0}^{N-1}g_N(i-j)\,a_i^\dagger a_j$ preserving the total fermion number. We shall
further assume that the hopping amplitude~$g_N$ satisfies $g_N(k)=g_N(-k)^*=g_N(k+N)$, so that~$H$
is Hermitian and translationally invariant. For this reason, it is convenient to introduce the
Fourier-transformed creation operators
\begin{equation}
  \label{hatas}
  \had_j=\frac1{\sqrt N}\sum_{l=0}^{N-1}\e^{2\pi\iu jl/N}a^\dagger_l\,,\qquad 0\le l\le N-1\,.
\end{equation}
It is straightforward to check that the operators~$\ha_j$, $\had_j$ satisfy the canonical
anticommutation relations (CAR), and that they diagonalize~$H$. In fact, we have
$ H=\sum_{l=0}^{N-1}\vep_N(l)\,\had_l\ha_l\,, $ with
$\vep_N(l)=\sum_{j=0}^{N-1}g_N(j)\e^{2\pi\iu jl/N}\,. $ It can be shown that the total momentum
operator~$P$ is also diagonal in this representation, namely $P=\sum_{l=0}^{N-1}p_l\,\had_l\ha_l$,
with $p_l=2\pi l/N\bmod 2\pi$. Thus the operator~$\had_l$ creates a (non-localized) fermion with
well-defined energy~$\vep_N(l)$ and momentum~$p_l$. Note that~$\vep_N(l)$ is obviously real for
all modes~$l$, and that the model is critical (gapless) if~$\vep_N(l)$ vanishes for some~$l$. We
shall suppose in what follows that the system is in a pure energy eigenstate
\begin{equation}
    \label{K}
    \ket K\equiv\had_{k_1}\cdots\,\had_{k_M}\ket0\,,\qquad K=\{k_1,\dots,k_M\}\subset\{0,\dots,N-1\}\,,
\end{equation}
where~$\ket0$ is the vacuum, consisting of~$M$ fermions with momenta $2\pi k_j/N$. We shall be
interested in studying the entanglement properties of a subset of sites
$ A\equiv\{x_1,\dots,x_L\}\subset\{0,\dots,N-1\} $ with respect to the whole system when the
latter is in the pure state~$\ket K$. As is well known, these properties are encoded in the
reduced density matrix $\rho_A=\tr_B\rho\,,$ where~$\rho\equiv\kbo K$ and~$B=\{0,\dots,N-1\}-A$.
As mentioned in the Introduction, the degree of entanglement is usually measured using the R\'enyi
entanglement entropy~$S_\al(A)\equiv (1-\al)^{-1}\log\tr(\rho_A^\al)$ (with~$\al>0$). One of the
most efficient ways of computing this entropy is to exploit the connection between the reduced
density matrix~$\rho_A$ and the correlation matrix~$C_A$, defined by
\begin{equation}
  \label{CAdef}
  (C_A)_{jk}=\ksp{K}{a_{x_j}^\dagger {a\vphantom{\vrule depth2.4pt}}_{x_k}}{K}\,,\qquad 1\le j,k\le L\,.
\end{equation}
This matrix is obviously Hermitian, with eigenvalues~$\nu_1,\dots,\nu_L$ lying in the
interval~$[0,1]$. Moreover, since the state~$\ket K$ is determined by the conditions
$\had_k\ket K=0$ for $k\in K$ and $\ha_k\ket K=0$ for $k\not\in K$, the expectation value
$\ksp{K}{\had_j\ha_k}{K}$ vanishes for $k\notin K$ and equals~$\de_{jk}$ for $k\in K$. From this
fact and Eq.~\eqref{hatas} we immediately obtain the following explicit expression for the matrix
elements of the correlation matrix~$C_A$:
\begin{equation}
  \label{CAexp}
  (C_A)_{jk}=\frac1N\sum_{l\in K}\e^{-2\pi\iu (x_j-x_k)l/N}\,,\qquad 1\le j,k\le L\,.
\end{equation}
As first shown in Refs.~\cite{Pe03,VLRK03}, the reduced density matrix~$\rho_A$ factors as the
tensor product $ \rho_A=\bigotimes_{l=1}^{L}\rho_A^{(l)}\,, $ where each~$\rho_A^{(l)}$ is a
$2\times2$ matrix with eigenvalues $\nu_l$ and $1-\nu_l$. In particular, the spectrum of~$\rho_A$
is the set of numbers
\begin{equation}\label{entspecpos}
  \rho_A(\vep_1,\dots,\vep_L)=\prod_{l=1}^L\big[\mkern2mu\nu_l^{\,\vep_l}
  (1-\nu_l)^{1-\vep_l}\mkern1mu\big]\,,
  \qquad\vep_l\in\{0,1\}\,.
\end{equation}
Since the R\'enyi entropy~$S_\al$ is additive, it follows that
\begin{equation}\label{SAexact}
  S_\al(A)=\sum_{l=1}^LS_\al\bigl(\rho_A^{(l)}\bigr)=(1-\al)^{-1}\sum_{l=1}^L\log\bigl(\nu_l^\al+(1-\nu_l)^\al\bigr)\,.
\end{equation}
Note that the latter method for computing~$S_\al(A)$ is computationally very advantageous, since
it is based on the diagonalization of the~$L\times L$ matrix~$C_A$ as opposed to direct
diagonalization of the~$2^L\times 2^L$ matrix~$\rho_A$.

As explained above, it is of great interest to determine the (leading) asymptotic behaviour of the
entanglement entropy~$S_\al(A)$ as the size~$L$ of the subsystem~$A$ tends to infinity. To this
end, note first of all that the matrix~$C_A$ is Toeplitz (i.e., $(C_A)_{jk}$ depends only on the
difference~$j-k$) provided that the subsystem~$A$ under consideration is a single block, i.e., a
set of consecutive sites. Let us further assume that Eq.~\eqref{CAexp} has a well-defined limit
as~$N\to\infty$ with~$L$ fixed, in the sense that there exists a piecewise smooth density
function~$c(p)$ such that $(C_A)_{jk}\to(2\pi)^{-1}\int_0^{2\pi}c(p)\,\e^{-\iu(j-k)p}\diff p$ in
this limit. As first shown by Jin and Korepin~\cite{JK04}, it is then possible to apply a
particular case of the Fisher--Hartwig conjecture~\cite{FH68} proved by Basor\cite{Ba79} to derive
an asymptotic formula for the characteristic polynomial of the correlation matrix~$C_A$, and hence
for the entanglement entropy~$S_\al(A)$ (see also Refs.~\cite{AEFS14,CFGRT16,CFGR17}). However,
when the subsystem~$A$ is not a single block it is clear from Eq.~\eqref{CAexp} that~$C_A$ is not
a Toeplitz matrix, and therefore the method just outlined cannot be used to derive the asymptotic
behaviour of~$S_\al(A)$ for large~$L$. It should also be stressed that the asymptotic result in
Ref.~\cite{JK04} is only valid for~$N\gg L\gg1$ (i.e., for an \emph{infinite} chain), since the
$N\to\infty$ limit with $L$ fixed is taken before letting~$L\to\infty$. In particular, the
asymptotic behaviour of~$S_\al(A)$ when $N\to\infty$ with $L/N\to\ga_x\in(0,1)$ cannot be directly
inferred from the latter result. As we shall explain shortly, these drawbacks can be overcome
through the use of a duality principle that we shall introduce below.

\subsection*{The dual correlation matrix}

We start by defining the projection of the operator~$\had_j$ onto the set~$\cL(\cH_A)$ of linear
operators from the Hilbert space~$\cH_A$ of the subsystem~$A$ into itself in the obvious way,
namely (cf.~Eq.~\eqref{hatas})
\begin{equation}
  \label{aA}
  \had_{A,j}=\frac1{\sqrt N}\sum_{l\in A}\e^{2\pi\iu jl/N}a^\dagger_l\,,
\end{equation}
and similarly for~$\ha_{A,j}$. We shall also denote by~$\had_{B,j}$, $\ha_{B,j}$ the corresponding
projections onto~$\cL(\cH_B)$, so that $\ha_j=\ha_{A,j}+\ha_{B,j}\,,$
$\had_j=\had_{A,j}+\had_{B,j}\,.$ We then define the \emph{dual correlation matrix}~$\hC_A$ as
the~$M\times M$ matrix with elements
\begin{equation}
  \label{hCA}
  (\hC_A)_{lm}=\Ksp0{{\ha\vphantom{\vrule depth2.85pt}}_{A,k_l}\had_{A,k_m}}0\,,\qquad 1\le l,m\le M\,.
\end{equation}
The dual correlation matrix~$\hC_B$ of the complementary set $B$ is defined similarly. The
analogue of the matrix~$\hC_A$ for continuous systems, usually called the overlap matrix, was
originally introduced by~Klich~\cite{Kl06} and has been extensively used in the literature (see,
e.g., Ref.~\cite{CMV11}). From the definition~\eqref{aA} of the projected operators~$\had_{A,j}$
we immediately obtain the explicit formula
\begin{equation}
  \label{hCAexp}
  (\hC_A)_{lm}=\frac1N\sum_{j\in A}\e^{-2\pi\iu(k_l-k_m)j/N}\,,\qquad 1\le l,m\le M\,.
\end{equation}
Comparison with Eq.~\eqref{CAexp} shows that $\hC_A$ is obtained from $C_A$ by exchanging the
roles played by the sites $x_j\in A$ and the excited modes~$k_l\in K$, which justifies the term
``dual correlation matrix''. We shall show in what follows that this duality can be successfully
exploited to obtain the asymptotic behaviour of~$S_\al(A)$ in situations in which the usual
approach based on the correlation matrix~$C_A$ is not feasible.

The matrix~$\hC_A$ is clearly Hermitian and positive semidefinite, since for
all~$z_1,\dots,z_M\in\CC$ we have
$ \sum_{l,m=1}^M(\hC_A)_{lm}\,z_l^*z_m=\big\Vert\big(\sum_{l=1}^Mz_m\,\had_{A,k_m}\big)
\ket0\big\Vert^2\,. $ Thus the eigenvalues~$\hnu_1,\dots,\hnu_M$ of~$\hC_A$ are nonnegative. Using
the identities $\ksp0{\cO_A\cO'_B}0=\ksp0{\cO'_B\cO_A}0=0\,,$ where~$\cO_A$ and~$\cO'_B$ are
linear operators respectively supported on~$A$ and $B$, it is straightforward to check that
$\hC_B=\id_M-\hC_A\,.$ Since~$\hC_B$ is also positive semidefinite, from the previous relation it
follows that~$\hnu_i\in[0,1]$ for all $i=1,\dots,M$. Moreover, the Hermitian character of~$\hC_A$
implies that there exists a unitary~$M\times M$ matrix~$U\equiv(u_{lm})_{1\le l,m\le M}$ such
that~$U\hC_AU^\dagger=\diag(\hnu_1,\dots,\hnu_M)$, and
hence~$U\hC_BU^\dagger=\id-U\hC_AU^\dagger=\diag(1-\hnu_1,\dots,1-\hnu_M)$. We then define the
corresponding rotated operators $\hc_l=\sum_{m=1}^Mu_{lm}\,\ha_{k_m}$ ($1\le l\le M$)\,, which
together with their adjoints satisfy the CAR by the unitarity of~$U$. We shall also need the
projections of the latter operators onto the spaces~$\cL(\cH_A)$ and~$\cL(\cH_B)$, namely
\begin{equation}\label{hcAB}
 \hc_{A,l}=\sum_{m=1}^Mu_{lm}\,\ha_{A,k_m}\,,\qquad \hc_{B,l}=\sum_{m=1}^Mu_{lm}\,\ha_{B,k_m}=\hc_l-\hc_{A,l}\,,
\end{equation}
and similarly for their adjoints. From the above definitions it follows that the vacuum
correlators of the operators $\{\hc_{A,l},\hcd_{A,l}\}$ and $\{\hc_{B,l},\hcd_{B,l}\}$ are given
by
\begin{equation}
  \label{hcABcorr}
  \Ksp0{\hc_{A,l}\hcd_{A,m}}0=\hnu_l\de_{lm}\,,\qquad \Ksp0{\hc_{B,l}\hcd_{B,m}}0=(1-\hnu_l)\de_{lm}\,,
\end{equation}
and hence $\big\Vert \hcd_{A,l}\ket0\big\Vert^2=\hnu_l$,
$\big\Vert \hcd_{B,l}\ket0\big\Vert^2=1-\hnu_l$\,. Following Ref.~\cite{Kl06}, we note that the
state $\ket\phi=\hcd_{1}\cdots\,\hcd_{M}\ket0$ actually differs from~$\ket K$ by an irrelevant
phase, since by definition of the operators~$\hc_l$ we have
\[
  \ket\phi=\sum_{m_1,\dots,m_M=1}^Mu_{1m_1}^*\cdots\,
  u_{M,m_M}^*\,\had_{k_{m_1}}\cdots\,\had_{k_{m_M}}\ket0 =\bigg(\sum_{\si\in S_M}(-1)^\si
  u_{1\si_1}^*\cdots\,\, u_{M,\si_M}^*\bigg)\had_{k_1}\cdots\,\had_{k_M}\ket0=\det U^*\,\ket K\,,
\]
where~$(-1)^\si$ denotes the sign of the permutation~$\si$. The latter relation implies
that~$\kbo K=\kbo\phi$, a fact that can be exploited in order to derive an expression for the
entanglement entropy~$S_\al(A)$. To this end, for~$\hnu_l\ne0,1$ we define the operators
$\hdd_{A,l}=\hcd_{A,l}/\sqrt{\hnu_l}$\,, $\hdd_{B,l}=\hcd_{B,l}/\sqrt{1-\hnu_l}$\,, so that by
Eq.~\eqref{hcABcorr} the states $\ket1_{A,l}\equiv\hdd_{A,l}\ket0$\,,
$\ket1_{B,l}\equiv\hdd_{B,l}\ket0$ are properly normalized. On the other hand, when~$\hnu_l=0$ the
state~$\hcd_l\ket0=\hcd_{B,l}\ket0$ is supported on~$B$ by Eq.~\eqref{hcABcorr}, and is
normalized, since the operators~$\hc_l,\hcd_l$ obey the CAR. Hence in this case we simply set
$\hdd_{B,l}=\hcd_{B,l}=\hcd_l$\,, $\ket1_{B,l}=\hdd_{B,l}\ket0$\,. Similarly, when~$\hnu_l=1$ we
define $\hdd_{A,l}=\hcd_{A,l}=\hcd_l$\,, $\ket1_{A,l}=\hdd_{A,l}\ket0$, and the previous
definitions we thus have $\hcd_l=\sqrt{\hnu_l}\;\hdd_{A,l}+\sqrt{1-\hnu_l}\;\hdd_{B,l}$
($1\le l\le M$)\,, and therefore
$\ket\phi=\bigotimes_{l=1}^M\Big(\sqrt{\hnu_l}\;\ket1_{A,l}\ket0_{B,l}+
\sqrt{1-\hnu_l}\;\ket0_{A,l}\ket1_{B,l}\Big)$, where~$\ket0_{A,l}$, $\ket0_{B,l}$ denote the
vacuum state in the~$l$-th mode (with respect to the~$\hcd_m$ operators) supported respectively
on~$A$ or~$B$. Using the identity~$\kbo K=\kbo\phi$ and tracing over the degrees of freedom of the
subsystem~$B$ we easily arrive at the fundamental formula
\begin{equation}
  \label{rhoAprod}
  \rho_A=\bigotimes_{l=1}^M\Big(\hnu_l\ket1_{A,l}\bra1_{A,l}+(1-\hnu_l)\ket0_{A,l}\bra0_{A,l}\Big)\,.
\end{equation}
In particular, the spectrum of the matrix~$\rho_A$ is the set of numbers
\begin{equation}\label{entspec}
  \rho_A(\vep_1,\dots,\vep_M)=\prod_{l=1}^M\big[\mkern2mu\hnu_l^{\,\vep_l}
  (1-\hnu_l)^{1-\vep_l}\mkern1mu\big]\,,
  \qquad\vep_l\in\{0,1\}\,,
\end{equation}
up to zero eigenvalues. From the additivity of the R\'enyi entropy and Eq.~\eqref{rhoAprod}
or~\eqref{entspec} it follows that the entanglement entropy~$S_\al(A)$ is given by
\begin{equation}\label{SAdual}
  S_\al(A)=(1-\al)^{-1}\sum_{l=1}^M\log\bigl(\,\hnu_l^{\,\al}+(1-\hnu_l)^\al\bigr)\,,
\end{equation}
which can be interpreted as the dual of Eq.~\eqref{SAexact}.

\subsection*{The duality principle}

As we have seen in the previous subsection, the R\'enyi entanglement entropy~$S_\al(A)$ can be
computed in two equivalent ways, using the ``coordinate'' correlation matrix~$C_A$ and its
``dual'' $\hC_A$ (cf.~Eqs.~\eqref{SAexact}-\eqref{SAdual}). This fact strongly suggests the
existence of a deeper duality principle that applies to the reduced density matrix~$\rho_A$
itself, as evidenced by Eqs.~\eqref{entspecpos}-\eqref{entspec}. To formulate this principle, we
shall introduce the more precise notation~$\rho_A(K)$ to denote the reduced density matrix of the
subsystem~$A$ when the whole system is in the pure energy eigenstate~$\ket K$ given by
Eq.~\eqref{K}. It should be borne in mind that in this notation~both sets~$A$ and~$K$ are subsets
of~$\{0,\dots,N-1\}$, with the subindex always labelling the subsystem sites (in position space)
and the argument the set of excited momenta. Let~$\spec T$ stand for the spectrum of the
matrix~$T$, i.e., the set of its eigenvalues, each counted with its respective multiplicity.
Likewise, we shall denote by $\spec_0\rho$ the spectrum of a density matrix~$\rho$ excluding its
zero eigenvalues, i.e., $\spec_0\rho=\spec\bigl(\rho|_{(\ker\rho)^\perp}\bigl)\,.$ We shall then
say that two density matrices~$\rho_{i}$ ($i=1,2$) \emph{are similar up to zero eigenvalues} if
$\spec_0\rho_1=\spec_0\rho_2$, i.e., $\rho_1$ and~$\rho_2$ have the same nonzero eigenvalues with
the same multiplicities. We are now ready to state the following fundamental result:
\begin{thm}\label{thm.thm}
  The reduced density matrices~$\rho_A(K)$ and~$\rho_K(A)$ are similar up to zero eigenvalues.
\end{thm}
\begin{proof}
  Indeed, by Eqs.~\eqref{entspecpos}-\eqref{entspec} the spectrum of~$\rho_A(K)$ excluding the
  zero eigenvalues can be written in the two equivalent ways
  \begin{equation}\label{essspec}
    \spec_0\bigl(\rho_A(K)\bigr)=\bigg\{\prod_{l=1}^L\nu_l^{\vep_l}(1-\nu_l)^{1-\vep_l}\mid\vep_l\in\{0,1\},\
    \nu_l\notin\{0,1\}\bigg\}
    =\bigg\{\prod_{m=1}^M\hnu_m^{\,\vep_m}(1-\hnu_m)^{1-\vep_m}\mid\vep_m\in\{0,1\},
    \ \hnu_m\notin\{0,1\}\bigg\}\,.
  \end{equation}
  Let us denote by~$C_A(K)$ and~$\hC_A(K)$ the correlation matrix~\eqref{CAexp} and its dual
  version~\eqref{hCAexp}. We then have~$\hC_A(K)=C_K(A)$, $C_A(K)=\hC_K(A)$, and consequently the
  sets $\{\nu_l\}_{l=1}^L$ and $\{\hnu_m\}_{m=1}^M$ are interchanged by the duality
  transformation~$A\lra K$\,. Applying Eq.~\eqref{essspec} to the reduced density
  matrix~$\rho_K(A)$ we conclude that~$\spec_0\bigl(\rho_A(K)\bigr)=\spec_0\bigl(\rho_K(A)\bigr)$,
  as claimed.
\end{proof}
If $S$ is any entropy functional, from now on we shall use the more precise notation
$S(A;K)=S\bigl(\rho_A(K)\bigr)$\,. Obviously, from the Shannon--Khinchin axioms it follows that
two density matrices which are similar up to zero eigenvalues necessarily have the same entropy.
From this fact and the previous theorem one can immediately deduce the important \emph{duality
  principle}
\begin{equation}
  \label{SAKKA}
  S(A;K)=S(K;A)\,,
\end{equation}
valid for any entropy functional~$S$.

As a first application of this general principle, we shall rigorously derive an asymptotic
expression for the R\'enyi entanglement entropy of a subsystem~$A$ consisting of~$r>1$ disjoint
blocks of consecutive spins when the set~$K$ of excited momenta is a single set of $M$ consecutive
integers, valid in the limit $N\gg M\gg 1$. More precisely, let $A=\bigcup_{i=1}^r[U_i,V_i)$,
$K=[P,Q)$\,, where~$[U_i,V_i)$ denotes the set of all integers~$l$ such that $U_i\le l<V_i$ (so
that the cardinal of~$[U_i,V_i)$ is $V_i-U_i$), and similarly for~$[P,Q)$. We first
let~$N\to\infty$ with~$M$ fixed and assume that the following limits exist:
\[
\lim_{N\to\infty}\frac{2\pi U_i}{N}\equiv u_i\,,\qquad \lim_{N\to\infty}\frac{2\pi V_i}{N}\equiv
v_i\,,
\]
with~$u_i,v_i\in[0,2\pi]$, $u_{i+1}-v_i>0$, $v_r-u_1<2\pi$\,. We shall be interested in the
asymptotic behavior of the R\'enyi entropy~$S_\al$ as $M\to\infty$. Thus the problem at hand is
precisely the dual of the one solved in Refs.~\cite{AFC09,CFGR17} with the help of the
Fisher--Hartwig conjecture. One of the main results of the latter references can be recast in the
present context as the asymptotic formula
\begin{equation}
  \label{S1s}
  S_\al\bigl([U,V);\tcup_{j=1}^s[P_j,Q_j)\bigr)\sim
  b_\al\bigg(s\log L
  +\sum_{j=1}^s\log\bigl(2\sin\bigl(\tfrac{q_j-p_j}2\bigr)\bigr)+\log f(\bp,\bq)\biggr)+s\,c_\al\,,
\end{equation}
where $p_j\equiv\lim\limits_{N\to\infty}(2\pi P_j/N)$,
$q_j\equiv\lim\limits_{N\to\infty}(2\pi Q_j/N)$,
\begin{equation}
  \label{fpqCal}
  b_\al=\frac16\bigg(1+\frac1\al\bigg),\en
  c_\al=\frac1{1-\al}\int_0^\infty\bigg(\al\csch^2t
  -\csch t\csch(t/\al)-\frac{1-\al^2}{6\al}\,\e^{-2t}\bigg)\frac{\diff t}t,\en
    f(\bp,\bq)=\prod_{1\le i<j\le s}\frac{\sin\bigl(\tfrac{q_j-p_i}2\bigr)
    \sin\bigl(\tfrac{p_j-q_i}2\bigr)}{\sin\bigl(\tfrac{p_j-p_i}2\bigr)
    \sin\bigl(\tfrac{q_j-q_i}2\bigr)}
\end{equation}
and the~$\sim$ notation means that the difference between the LHS and the RHS tends to $0$ as
$L\to\infty$. From the duality relation~\eqref{SAKKA} and Eqs.~\eqref{S1s}-\eqref{fpqCal} it then
follows that when~$M\to\infty$ we have
\begin{equation}
  S_\al\bigl(\tcup_{i=1}^r[U_i,V_i);[P,Q)\bigr)
  =
  S_\al\bigl([P,Q);\tcup_{i=1}^r[U_i,V_i)\bigr)
  \sim b_\al\bigg[r\log M
  +\sum_{i=1}^r\log\bigl(2\sin\bigl(\tfrac{v_i-u_i}2\bigr)\bigr)+\log f(\bu,\bv)\biggr]
  +r\,c_\al\,.
    \label{Salprev}
\end{equation}
Taking into account that~$f(\bu,\bv)=1$ when~$r=1$, from the previous formula we obtain the
remarkable formula
\begin{equation}
  \label{Sr1MI}
  S_\al\bigl(\tcup_{i=1}^r[U_i,V_i);[P,Q)\bigr)\sim
  \sum_{i=1}^rS_\al\bigl([U_i,V_i);[P,Q)\bigr)-I_\al(\bu,\bv)\,,\qquad\text{with}\quad
  I_\al(\bu,\bv)\equiv-b_\al\log f(\bu,\bv)\,,
\end{equation}
where the last term can be naturally interpreted as an asymptotic approximation to the
\emph{mutual information} shared by the blocks $[U_1,V_1)\,,\dots,[U_r,V_r)$. We believe that this
is the first time that this asymptotic formula, which agrees with well-known CFT results, has been
rigorously established using the (proved part of the) Fisher--Hartwig conjecture.

It is important to keep in mind the limiting process leading to Eq.~\eqref{Salprev} in order to
correctly assess its limit of validity. For instance, using the connection between one-dimensional
critical systems and $1+1$ dimensional CFTs it follows that the
asymptotic behavior of~$S_\al$ is given (in our notation) by~\cite{CFH05,AEF14}
\begin{equation}\label{SalFalc}
  S_\al\bigl(\tcup_{i=1}^r[U_i,V_i);[P,Q)\bigr)\sim
  b_\al\bigg[r\log\biggl(\frac N\pi\sin\biggl(\pi\,\frac MN\biggr)\biggr)
  +\sum_{i=1}^r\log(v_i-u_i)+\log \finf(\bu,\bv)\biggr]+r\,c_\al\,,
\end{equation}
where~$\finf(\bu,\bv)$ is the product of cross ratios
\begin{equation}\label{fFalc}
  \finf(\bu,\bv)=\prod_{1\le i<j\le r}\frac{(v_j-u_i)(u_j-v_i)}{(u_j-u_i)(v_j-v_i)}\,.
\end{equation}
The apparent discrepancy between the latter formulas and Eqs.~\eqref{fpqCal}-\eqref{Salprev} is
easily explained taking into account that the limiting process in the latter reference is the
\emph{dual} of the present one, namely~$N\to\infty$ with \emph{fixed} $U_i,V_i$ and
$2\pi P/N\to p$, $2\pi Q/N\to q$. In other words, Eqs.~\eqref{fpqCal}-\eqref{Salprev} apply when
$N\gg M\gg 1$ and arbitrary $L<N$, while Eqs.~$\eqref{SalFalc}$-\eqref{fFalc} are valid
for~$N\gg L\gg 1$ and arbitrary $M<N$. It is also obvious that both approaches coincide in the
(rather uninteresting) case in which both $M/N$ and $L/N$ tend to zero. On the other hand, it
should be apparent that neither Eqs.~\eqref{fpqCal}-\eqref{Salprev}
nor~\eqref{SalFalc}-\eqref{fFalc} are valid in the general situation in which \emph{both} $L/N$
and $M/N$ tend to a nonzero limit as $N\to\infty$. In fact, it is clear a priori that none of
these formulas can hold in the latter range, since they are not consistent with the invariance
under complements identity $S(A;K)=S(A^c;K)$ and its dual consequence
$S(A;K)=S(A;K^c)$, where~$A^c$ and $K^c$ respectively denote the complements of $A$
and~$K$ with respect to the set~$\{0,\dots,N-1\}$.

Our next objective is to find an extension of Eqs.~\eqref{Salprev} and~\eqref{SalFalc} valid in
the general case in which~both $\ga_x$ and~$\ga_p\equiv\lim_{N\to\infty}M/N$ tend to a nonzero
limit as~$N\to\infty$. To this end, consider first the simplest case in which $r=s=1$. By
translation invariance and criticality, as $N\to\infty$ we must have
$S_\al([U,V);[P,Q))\sim b_\al\log N+\si_\al(\ga_x,\ga_p)$, where $\ga_x=(V-U)/N$, $\ga_p=(Q-P)/N$
and~$\si_\al$ satisfies: i)~$\ds\si_\al(\ga_x,\ga_p)=\si_\al(\ga_p,\ga_x)$ (on account of the
duality principle~\eqref{SAKKA}),\en ii)~$\si_\al(\ga_x,\ga_p)=\si_\al(1-\ga_x,\ga_p)$ (by the
invariance of the entropy under complements),\en
iii)~$\si_\al(\ga_x,\ga_p)=b_\al\log\bigl(2\ga_x\sin(\pi\ga_p)\bigr)+c_\al+o(1)$\,,
with~$\lim_{\ga_x\to0}o(1)=0$ (by Eq.~\eqref{SalFalc} with $r=1$). (In fact, combining
conditions~i) with ii) and iii) it immediately follows
that~$\si_\al(\ga_x,\ga_p)=\si_\al(\ga_x,1-\ga_p)$
and~$\si_\al(\ga_x,\ga_p) =b_\al\log\bigl(2\ga_p\sin(\pi\ga_x)\bigr)+c_\al+o(1)$, where~$o(1)\to0$
as~$\ga_p\to0$\,.) Obviously, the simplest function satisfying the previous requirements
is~$\si_\al(\ga_x,\ga_p)=b_\al\log\bigl(\tfrac2\pi\sin(\pi\ga_x)\sin(\pi\ga_p)\bigr)+c_\al$,
obtained from~Eq.~\eqref{SalFalc} with~$r=1$ by the replacement~$\pi\ga_x\mapsto\sin(\pi\ga_x)$.
Numerical calculations show that for all~$\al>0$ the correct asymptotic formula
for~$S_\al([U;V);[P,Q))$ is indeed the simplest one, namely
\begin{equation}
  \label{Sal1+1}
  S_\al([U,V);[P,Q))\sim b_\al\log\biggl(\frac{2N}\pi\,\sin(\pi\ga_x)\sin(\pi\ga_p)\biggr)+c_\al
\end{equation}
(see, e.g., Fig.~\ref{fig.err1} (a) for the most ``unfavourable'' case~$\ga_x=\ga_p=1/2$). This
conclusion is also in agreement with the analogous result in Ref.~\cite{CMV11jstat} for the $XX$
model. In fact, we found the leading correction to the approximation~\eqref{Sal1+1} to be
monotonic in~$N$ and~$O(N^{-2})$ for $\al=1$, and $O\bigl(\cos(2\pi\ga_x\ga_pN\bigr)\,N^{-2/\al})$
for~$\al>1$ (cf.~Fig.~\ref{fig.err1}). This behaviour qualitatively agrees with the results of
Ref.~\cite{CE10} for the error of the Jin--Korepin asymptotic formula for the R\'enyi entanglement
entropy of the ground state of the \emph{infinite}~$XX$ chain (Eq.~\eqref{Sal1+1}
with~$\sin(\pi\ga_x)$ replaced by~$\pi\ga_x$). On the other hand, in the case~$0<\al<1$ (which was
not addressed in the latter reference), our numerical calculations suggest that the correction to
Eq.~\eqref{Sal1+1} is monotonic and $O(N^{-2})$.
\begin{figure}[ht]%
  \includegraphics[width=\linewidth]{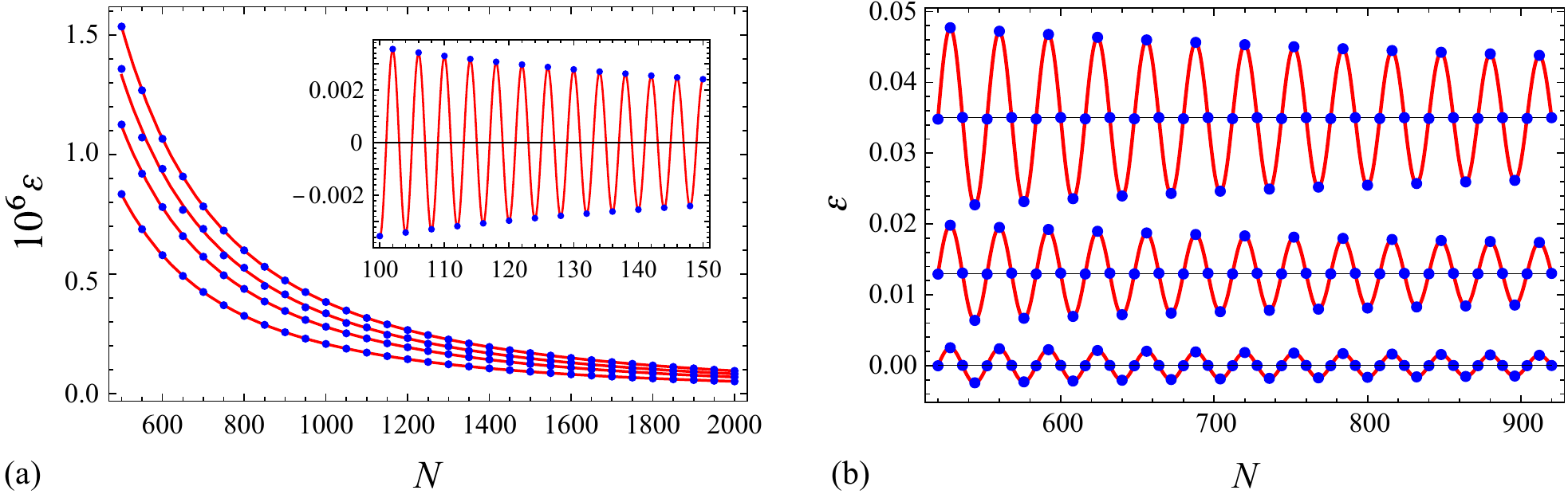}
  \caption{Difference~$\vep$ between the exact value of the R\'enyi entropy, computed via
    Eq.~\eqref{SAexact} by numerical diagonalization of the correlation matrix~\eqref{CAexp}, and
    its asymptotic approximation \eqref{Sal1+1} for (a)~$\ga_x=\ga_p=1/2$ and (b)~$\ga_x=1/8$,
    $\ga_p=1/4$. In panel (a) we have shown the cases (bottom to top) $\al=3/5$, $2/3$, $3/4$, $1$
    (von Neumann entropy) and $\al=2$ (inset), while panel (b) depicts the cases~$\al=2$, $5/2$,
    $3$ (bottom to top, with the horizontal axis displaced respectively by $0.013$ and $0.035$ in
    the last two cases to avoid overlap). The solid red lines represent the curves providing the
    best fits of the data to the laws~$a\,N^{-2}$ (main panel (a)) and
    $a\,N^{-2/\al}\cos(2\pi\ga_x\ga_p N)$ (inset of panel (a) and panel (b)).}
  \label{fig.err1}
\end{figure}

At this point, it is very natural to assume that Eq.~\eqref{Sr1MI}~and its dual are valid for
\emph{all} values of the parameters~$\ga_x,\ga_p\in(0,1)$, and not just for~$\ga_p\ll1$
or~$\ga_x\ll1$, respectively. The latter assumption and Eq.~\eqref{Sal1+1} thus lead to the
asymptotic formulas
\begin{align}\label{Salr1}
  S_\al\bigl(\tcup_{i=1}^r[U_i,V_i);[P,Q)\bigr)
  &\sim b_\al\bigg[r\log\biggl(\frac{2N}\pi\,\sin(\pi\ga_p)\biggr)
    +\sum_{i=1}^r\log\sin\bigl(\tfrac{v_i-u_i}2\bigr)\biggr]-I_\al(\bu,\bv)+r\,c_\al\,,\\
  S_\al\bigl([U,V);\tcup_{i=1}^s[P_i,Q_i)\bigr)
  &\sim b_\al\bigg[s\log\biggl(\frac{2N}\pi\,\sin(\pi\ga_x)\biggr)
    +\sum_{i=1}^s\log\sin\bigl(\tfrac{q_i-p_i}2\bigr)\biggr]-I_\al(\bp,\bq)+s\,c_\al\,.
    \label{Sal1s}
\end{align}
In fact, the validity of the latter equations can be established by noting that one can go from
Eq.~\eqref{S1s}, which holds for an \emph{infinite} chain, to its analogue for a finite chain by
the usual procedure~\cite{CMV11jstat,CTT14} of replacing the ``arc distance''~$L$ by the chord
length~$(N/\pi)\sin(\pi L/N)=(N/\pi)\sin(\pi\ga_x)$\,. In this way Eq.~\eqref{S1s} immediately
yields Eq.~\eqref{Sal1s}, which implies its counterpart~\eqref{Salr1} by the duality
principle~\eqref{SAKKA}.

Again, our numerical calculations for several block configurations and a wide range of values of
the R\'enyi parameter~$\al$ fully corroborate the validity of Eqs.~\eqref{Salr1}-\eqref{Sal1s} (see,
e.g., Fig.~\ref{fig.errr+1}). More precisely, our numerical analysis suggests that for
sufficiently large~$N$ the error term in the latter equations behaves
as~$f(N)O(N^{-\min(2,2/\al)})$, where~$f(N)$ is a periodic function of~$N$. In particular, the
error term may not be monotonic in~$N$ even for~$\al\le 1$, in contrast with what happens in the
$r=s=1$ case. The above results are in agreement with those reported in Ref.~\cite{FC10} for the
(infinite) $XY$ chain and its corresponding free fermion model with $\al>1$, $r=2$ and $s=1$.
\begin{figure}[ht]%
  \includegraphics[width=\linewidth]{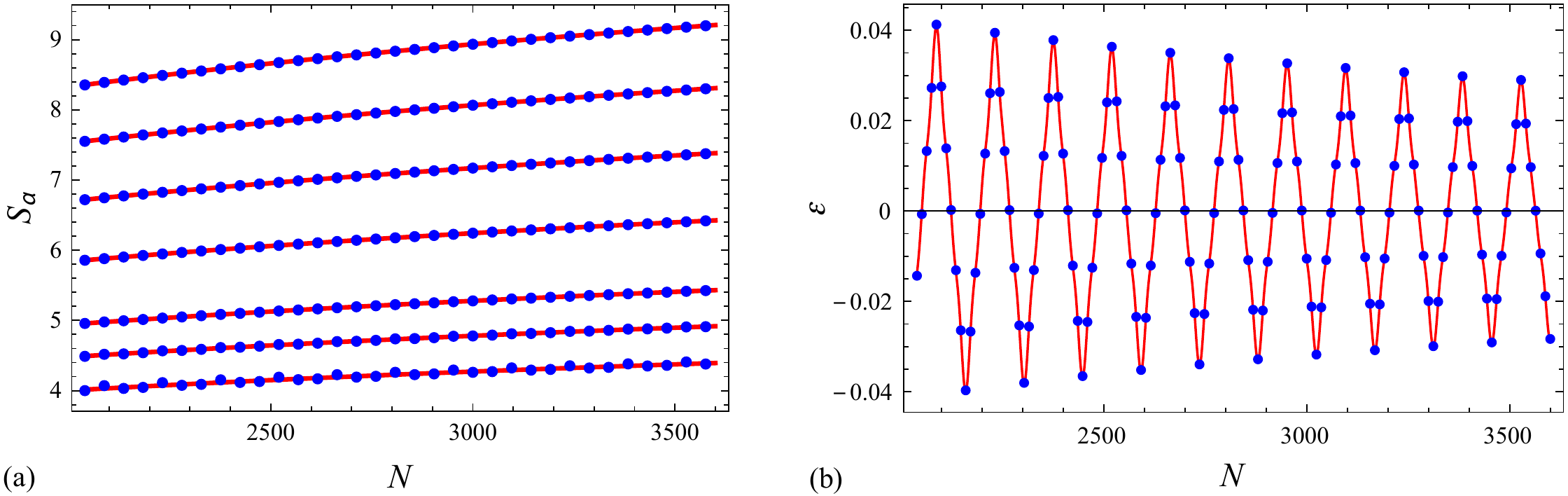}
  \caption{(a) Exact R\'enyi entropy~$S_\al$ (blue dots) vs.~its asymptotic
    approximation~\eqref{Salr1} (continuous red line) for a subsystem consisting of three
    equispaced blocks of equal length~$N/12$ when the whole system's state~\eqref{K} is made up of
    a sequence of consecutive excited modes of length~$N/12$ ($r=3$, $s=1$, $\ga_x=1/4$,
    $\ga_p=1/12$). The values of the R\'enyi parameter~$\al$ considered are (from top to
    bottom)~$1/2$, $3/5$, $3/4$, $1$, $3/2$, $2$ and $3$. (b) Difference~$\vep$ between the exact
    entropy~$S_3$ and its approximation~\eqref{Salr1} in the previous configuration as a function
    of the number of fermions~$N$. The continuous red line is the graph of the
    function~$f(N)N^{-2/3}$, with
    $f(N)=-5.54238\,\cos(\nu_0N)-0.742586\,\cos(3\nu_0N)-0.39794\,\cos(5\nu_0N)$ and
    $\nu_0=2\pi\ga_x\ga_p/r=\pi/72$.}
  \label{fig.errr+1}
\end{figure}

\subsection*{Multi-block entanglement entropy: conjecture for the general case}

We shall address in this section the general problem, in which both sets~$A$ and~$K$ consist of
several blocks of consecutive sites or modes, respectively. To the best of our knowledge, an
asymptotic formula for the entanglement entropy in this case has not previously appeared in the
literature. As explained above, the main difficulty is now that neither the correlation
matrix~$C_A$ nor its dual~$\hC_A$ are Toeplitz, so that the standard procedure based on the use of
the Fisher--Hartwig conjecture to obtain an asymptotic formula for the characteristic polynomial
of the correlation matrix~$C_A$ (or of its dual~$\hC_A$) is not applicable. Our approach for
deriving a plausible conjecture for the asymptotic behavior of $S_\al$ in the general case
considered in this subsection relies instead on the general duality principle discussed in the
previous section (cf.~Theorem~\ref{thm.thm} and~Eq.~\eqref{SAKKA}). In addition, we shall make the
natural assumption that when the distance between any two consecutive blocks $A_i$, $A_{i+1}$ is
much larger than the maximum block length (i.e., when
$\min_{1\le i\le r}(u_{i+1}-v_i)\gg\max_{1\le i\le r}(v_i-u_i)$, where~$u_{r+1}\equiv u_1+2\pi$)
the entanglement entropy is asymptotic to the sum of the single block entropies~$S_\al(A_i;K)$.
The motivation behind this assumption is that when the blocks are far apart their mutual influence
should be negligible, and the R\'enyi entropy is of course additive over independent events.

The simplest asymptotic formula satisfying the above assumption is the trivial one
$S_\al(A;K)\sim\sum_{i=1}^rS_\al(A_i;K)$. However, the latter formula cannot be correct, since it
violates the duality principle. The obvious way of fixing this shortcoming would be to add the
dual term $\sum_{j=1}^sS_\al(A;K_j)$ to the RHS, but the resulting formula violates the above
assumption. On the other hand, since by Eq.~\eqref{Sr1MI}
$\sum_{j=1}^s S_\al(A;K_j)\sim\sum_{i=1}^r\sum_{j=1}^s S_\al(A_i;K_j)-sI_\al(\bu,\bv)$,
and~$I_\al(\bu,\bv)\sim0$ when the blocks in coordinate space are far apart, the heuristic formula
\begin{equation}
  \label{genconj}
  S_\al(A;K)
  \sim\sum_{i=1}^rS_\al(A_i;K)+
  \sum_{j=1}^s
  S_\al(A;K_j)-\sum_{i=1}^r\sum_{j=1}^s
  S_\al(A_i;K_j)
\end{equation}
satisfies the above fundamental assumption. This relation is also clearly consistent with the
duality principle~\eqref{SAKKA}, since the RHS of Eq.~\eqref{genconj} is invariant under the
exchange of the sets $A$ and $K$ on account of Theorem~\ref{thm.thm}. We are thus led to
conjecture that when~$N\to\infty$ the R\'enyi entropy of a configuration with~$r$ blocks $A_i$ in
coordinate and $s$ blocks $K_j$ in momentum space satisfies the previous relation. Using
Eqs.~\eqref{Sr1MI}, its dual and Eq.~\eqref{Sal1+1} we immediately arrive at the closed asymptotic
formula
\begin{equation}
  S_\al(A;K)\sim rs\bigg(b_\al\log\bigg(\frac{2N}\pi\bigg)+c_\al\bigg)
    +s\bigg(b_\al\sum_{i=1}^r\log\sin\big(\tfrac{v_i-u_i}2\big)-I_\al(\bu,\bv)\bigg)
    +r\bigg(b_\al\sum_{i=1}^s\log\sin\big(\tfrac{q_i-p_i}2\big)-I_\al(\bp,\bq)\bigg)\,.
  \label{conjfinal}
\end{equation}
The latter equation is manifestly consistent with the duality principle stated in
Theorem~\ref{thm.thm}, as expected from the previous remark. It is also apparent that
Eq.~\eqref{conjfinal} reduces to Eq.~\eqref{Salr1} or~\eqref{Sal1s} respectively for $s=1$ or
$r=1$, as the asymptotic mutual information~$I_\al$ vanishes for a single block. Moreover, it is
straightforward to explicitly check that when the blocks in coordinate space are far apart the RHS
reduces to the sum of the asymptotic approximations~\eqref{Sal1s} to the single-block
entropies~$S_\al(A_i;K)$, since~$I_\al(\bu,\bv)\sim0$~in this limit. (By duality, a similar remark
applies to the case in which the blocks~$[P_j,Q_j)$ in momentum space are far apart from each
other.) Finally, it is immediate to check that Eq.~\eqref{conjfinal} satisfies the invariance
under complements identity.
\begin{figure}[t]%
  \includegraphics[width=\linewidth]{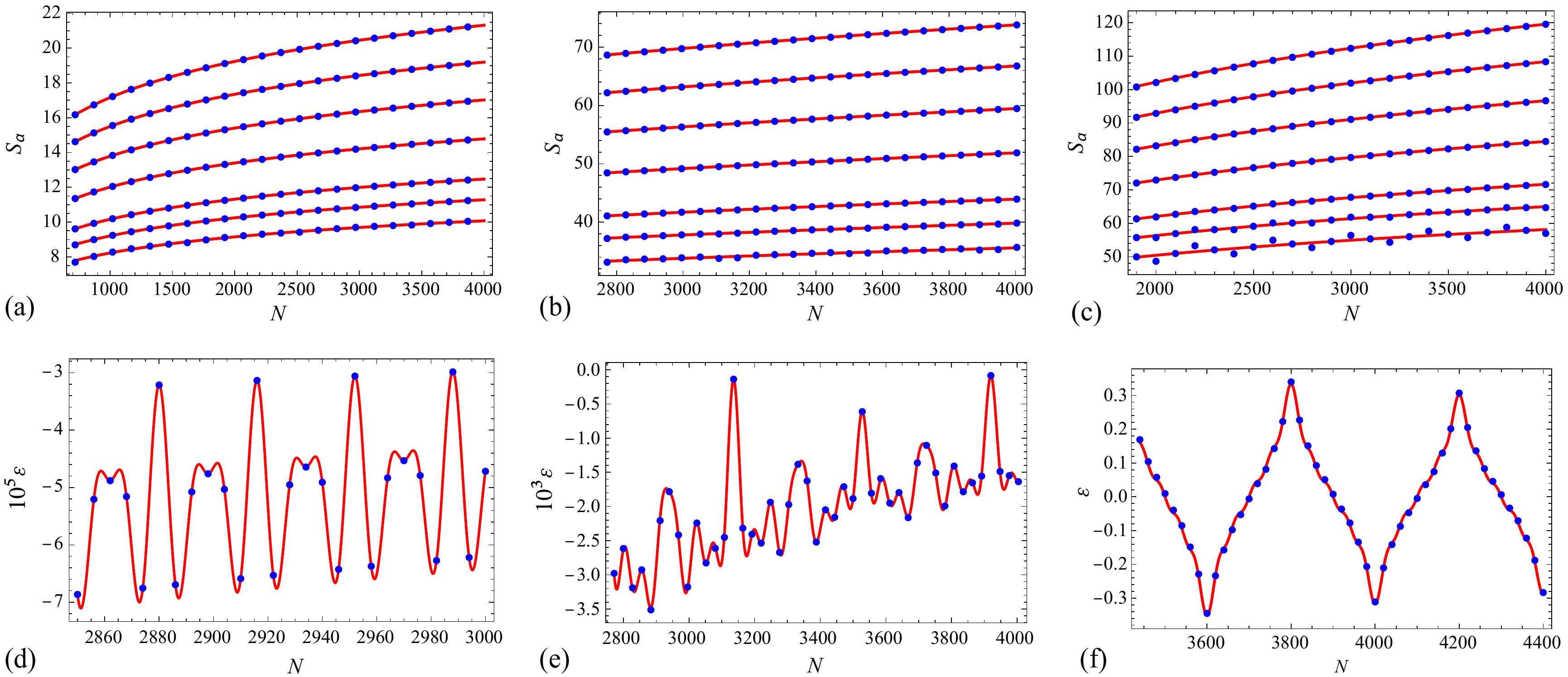}
  \caption{(a)--(c) Exact R\'enyi entropy $S_\al$ (blue dots) and its asymptotic
    approximation~\eqref{conjfinal} (continuous red line) for~$\al=1/2$, $3/5$, $3/4$, $1$, $3/2$,
    $2$, $3$ (top to bottom) in (a) a symmetric configuration with $r=3$, $s=2$, $\ga_x=1/2$,
    $\ga_p=1/3$, (b) an asymmetric configuration with $r=7$, $s=4$, $\ga_x=1/2$, $\ga_p=1/4$, and
    (c) a symmetric configuration with $r=10$, $s=5$, $\ga_x=1/2$, $\ga_p=1/4$. (d)--(f)
    Difference~$\vep$ between the exact entropy $S_\al$ and its approximation~\eqref{conjfinal}
    for the above configurations and (d)~$\al=1/2$, (e)~$\al=1$ (von Neumann entropy), and
    (f)~$\al=2$. The red lines represent the corresponding curves~$f(N)N^{-\min(2,2/\al)}$,
    with~$f(N)=\sum_{k=0}^{k_{\mathrm{max}}}a_k\cos(k\nu N)$ given in Table~\ref{tab.err}.}
\label{fig.entrs}
\end{figure}
\begin{figure}[b!]%
  \includegraphics[width=\linewidth]{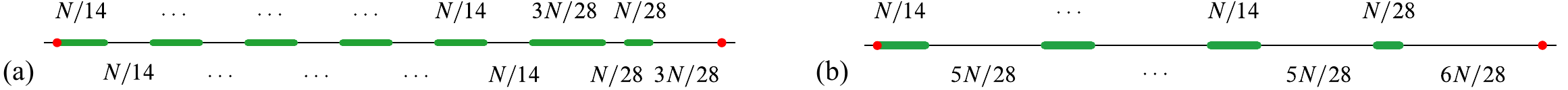}
  \caption{Asymmetric block configuration discussed in Fig.~\ref{fig.entrs} (b) in (a) coordinate
    space, (b) momentum space (the thick green lines represent the blocks, and the red dots are
    the two identified endpoints of the chain).}
  \label{fig.configs}
\end{figure}
We have verified through extensive numerical calculations with a wide range of configurations in
coordinate and momentum space that when~$N\gg1$ Eq.~\eqref{conjfinal} is correct. In fact, for
symmetric configurations (consisting of equally spaced blocks of the same length, both in
coordinate and momentum space) the error term in the latter equation behaves
as~$f(N)N^{-\min(2,2/\al)}$, where~$f$ is again a periodic function. More precisely (for
rational~$\ga_x$ and $\ga_p$), $f(N)$ is well approximated by a trigonometric
polynomial~$\sum_{k=0}^{k_{\mathrm{max}}}a_k\cos(k\nu N)$ with small $k_{\mathrm{max}}$
(independent of~$N$), where the main frequency $\nu$ is the product
of~$\nu_0\equiv2\pi\ga_x\ga_p/rs$ with a simple fraction that can be computed from the
configuration parameters $r$, $s$, $\ga_x$, $\ga_p$. The behavior of the error is very similar in
non-symmetric configurations, except that in some cases it appears to decay faster than~$N^{-2}$
for $0<\al<1$. As an example, in Fig.~\ref{fig.entrs} we present our results for three different
configurations with $(r,s)=(3,2),(7,4),(10,5)$. More precisely, the first and last of these
configurations are symmetric, while the middle one is (slightly) asymmetric, as detailed in
Fig.~\ref{fig.configs}. As can be seen from Figs.~\ref{fig.entrs} (d)-(f), the error in
Eq.~\eqref{conjfinal} behaves in these three cases as described above, where the coefficients
$a_k$ of the trigonometric polynomial~$f(N)$ and its fundamental frequency~$\nu$ are listed in
Table~\ref{tab.err}.
\begin{table}[ht]
\centering
\begin{tabular}{|c|c|l|c|c|}
  \hline
  Case & $k_{\mathrm{max}}$ & \hfill $(a_0,\dots,a_{k_{\mathrm{max}}})$\hfill\null& $\nu_0$ &$\nu$\\
  \hline
  (d) & $2$ & $(-438.485,105.29,66.716)$& $\pi/18$&$\nu_0$\\\hline
  (e) & $14$ &
               \begin{minipage}[t]{.6\linewidth}
                 $(-21790.1,76.0009,1602.85,154.097,5143.99,397.121,416.007,1950.55,$\\
                 $4556.52,156.444,756.382,168.572,2164.74,232.817,2661.63)$
               \end{minipage}
       &  $\pi/112$& $2\nu_0/7$\\\hline
  (f)& $9$& $(0,-852.969,0,-202.359,0,-99.4396,0,-57.2755,0,-55.2294)$
                                                                              &$\pi/200$& $\nu_0$\\
  \hline
\end{tabular}
\caption{\label{tab.err}Coefficients~$a_k$ and fundamental frequency~$\nu$ of the trigonometric
  polynomial~$f(N)=\sum_{k=0}^{k_{\mathrm{max}}}a_k\cos(k\nu N)$ in the error of
  Eq.~\eqref{conjfinal} for cases~(d)-(f) in Fig.~\ref{fig.entrs}.}
\end{table}

It should be noted that the asymptotic formula~\eqref{conjfinal}, which we have numerically
checked for a \emph{finite} chain, easily yields as a limiting case an analogous formula for an
infinite chain. Indeed, if in Eq.~\eqref{conjfinal} we let $\ga_x$ tend to $0$ we have
$\sin\bigl((v_i-u_i)/2\bigr)\simeq \pi(V_i-U_i)/N$, and similarly for the other arguments of the
sine functions appearing in the asymptotic mutual information term~$I_\al(\bu,\bv)$. In this way
we easily arrive at the analogue of~Eq.~\eqref{conjfinal} for an infinite chain, namely
\begin{equation}\label{Salinf}
  S_\al^{(\infty)}
  \sim s\,b_\al\log\bigg[\prod_{i=1}^r(V_i-U_i)\cdot\prod_{1\le i<j\le r}
  \frac{(V_j-U_i)(U_j-V_i)}{(U_j-U_i)(V_j-V_i)}\bigg]
  +r\bigg(b_\al\sum_{i=1}^s\log\sin\big(\tfrac{q_i-p_i}2\big)-I_\al(\bp,\bq)\bigg)
  +rs(b_\al\log2+c_\al)\,.
\end{equation}
To the best of our knowledge, this general asymptotic formula has not previously appeared in the
literature. Note also that for~$s=1$ (i.e., when there is a single block of excited momenta)
Eq.~\eqref{Salinf} implies the asymptotic expression for the mutual information of~$r$ blocks in
coordinate space conjectured in Ref.~\cite{AEF14}.

From the asymptotic approximation~\eqref{conjfinal} (or its equivalent version
Eq.~\eqref{genconj}) one can also deduce a remarkable expression for the (asymptotic) mutual
information of $r$ blocks~$A_i\equiv[U_i,V_i)$ ($1\le i\le r$) in position space when the chain is
in an energy eigenstate~$\ket K$ made up of $s$ blocks~$K_j\equiv[P_j,Q_j)$ ($1\le j\le s$) of
excited momentum modes, defined as
$ \cI_\al\bigl(A_1,\dots,A_r;K\bigr) \equiv\sum_{i=1}^rS_\al\bigl(A_i;K\bigr)-
S_\al\bigl(\tcup_{i=1}^rA_i;K\bigr)\,. $ Indeed, using Eqs.~\eqref{Sr1MI} and~\eqref{genconj} we
immediately obtain the asymptotic formula
\begin{equation}\label{cIapp}
  \cI_\al\bigl(A_1,\dots,A_r;K\bigr)
  \sim\sum_{j=1}^s\bigg[\sum_{i=1}^r
  S_\al\bigl(A_i;K_j\bigr)-S_\al\bigl(\tcup_{i=1}^rA_i;K_j\bigr)\bigg]
  \sim\sum_{j=1}^sI_\al(\bu,\bv)=s\,I_\al(\bu,\bv)\,.
\end{equation}
Thus (in the large~$N$ limit) the multi-block mutual information~$\cI_\al$ is simply~$s$ times the
mutual information when the chain's state~$\ket K$ consists of a single block of consecutive
momenta. In particular, we see that~$\cI_\al$ depends only on the \emph{topology} of the
state~$\ket K$ (i.e., the number of blocks of excited momenta), not on its \emph{geometry} (i.e.,
the particular arrangement and the lengths of these blocks). One could also define the mutual
information of~$s$ blocks of excited momenta~$K_j\equiv[P_j,Q_j)$ ($1\le j\le s$) for a fixed
configuration~$A\equiv\bigcup_{i=1}^r A_i$ in position space. It easily follows from
Eq.~\eqref{cIapp} and the duality principle that this mutual information is asymptotic
to~$r I_\al(\bp,\bq)$. Of course, an analogous formula should hold for the infinite chain
replacing the function $I_\al$ by its $N\to\infty$ limit
$I_\al^{(\infty)}(\bU,\bV)=b_\al\log\finf(\bU,\bV)\,.$ In particular, for $s=1$ the latter
expression implies that the model-dependent overall factor appearing in the general formula for
the mutual information of a $1+1$ dimensional CFT~(see, e.g., Refs.~\cite{CFH05,CC09,CTT14}) is
equal to $1$ for the models under consideration.

An alternative measure of the information shared by the blocks~$A_i$ ($1\le i\le r$) discussed in
Ref.~\cite{CTT14} is the quantity
$\widetilde\cI_\al(A_1,\dots,A_r)\equiv\sum_{l=1}^r(-1)^{l+1}\sum_{1\le i_1<\dots<i_l\le
  r}S_\al(\bigcup_{k=1}^lA_{i_k})$ (we omit the dependence on the chain's state~$\ket K$ for
conciseness's sake). In particular, for~$r=3$ we obtain the \emph{tripartite information}
introduced in Ref.~\cite{CH09}, whose vanishing characterizes the extensivity of the mutual
information~$\cI_\al$. It can be readily checked that the asymptotic relation~\eqref{conjfinal}
implies that~$\widetilde\cI_\al(A_1,\dots,A_r)$ vanishes asymptotically for the models under
consideration. This follows immediately from Eq.~\eqref{cIapp} ---which is itself a consequence
of~\eqref{conjfinal}--- and the
identities~$\sum_{1\le i_1<\dots<i_l\le
  r}\sum_{k=1}^lS_\al(A_{i_k})=\binom{r-1}{l-1}\sum_{i=1}^rS_\al(A_i)$,
$\sum_{1\le i_1<\dots<i_l\le
  r}I_\al\bigl((u_{i_1},\dots,u_{i_l}),(v_{i_1},\dots,v_{i_l})\bigr)=\binom{r-2}{l-2}I_\al(\bu,\bv)$.
In particular, this shows that the conjecture~\eqref{conjfinal} implies the asymptotic extensivity
of the mutual information~$\cI_\al$ for the models under consideration. (For the infinite chain
with $s=1$, this had already been noted in Ref.~\cite{AEF14}.)

Another noteworthy consequence of the asymptotic formula~\eqref{conjfinal} is the fact that for
large~$N$ the entanglement entropy can be approximately written as (omitting, for simplicity, its
arguments)
\begin{equation}\label{Sstruct}
  S_\al\sim rs\bigg(b_\al\log\biggl(\frac{2N}\pi\biggr)+c_\al\bigg)+b_\al g\,,\quad
  \text{with}\quad g\equiv
  s\bigg(\sum_{i=1}^r\log\sin\big(\tfrac{v_i-u_i}2\big)+\log f(\bu,\bv)\bigg)
  +r\bigg(\sum_{i=1}^s\log\sin\big(\tfrac{q_i-p_i}2\big)+\log f(\bp,\bq)\bigg)\,.
\end{equation}
The term in parenthesis in the latter formula, which contains the leading
contribution~$rsb_\al\log N$ to~$S_\al$ as $N\to\infty$, depends only on the topology of the
configuration considered. In particular, from the coefficient of the $\log N$ term we deduce that
the models under consideration are critical, behaving as a $1+1$ dimensional CFT with central
charge~$rs$. Note also that the fact that the leading asymptotic behavior of the R\'enyi
entanglement entropy~$S_\al$ depends only on the topology of the configuration in \emph{both}
position and momentum space is a generalization of the widespread hypothesis (for the case~$r=1$)
that the entanglement properties of critical fermion models are determined by the topology of
their Fermi ``surface'' (see, e.g.,~Ref.~\cite{ECP10}).

On the other hand, the numerical constant~$g$ in the previous equation is independent
of~$N$ and~$\al$, and is solely determined by the geometry of the configuration in both position
and momentum space. For instance, for the two symmetric configurations discussed
in~Fig.~\ref{fig.entrs}~(a), (c) this constant is respectively equal to~$-3\log 12$ and
$-25\log 1250$.

The asymptotic formula~\eqref{Sstruct} makes it possible to tackle several relevant problems that
would otherwise be intractable in practice. For instance, it is natural to conjecture that fixing
$r$, $s$, $\ga_x$ and $\ga_p$ the block configuration which maximizes the entropy is the symmetric
one (i.e., $r$ equally spaced blocks of equal length in position space, and similarly in momentum
space). Our numeric calculations for several configurations suggest that this is indeed the case
(see, e.g., Fig.~\ref{fig.entcmp} (a) for the case~$\al=2$). As we see from Eq.~\eqref{Sstruct},
this problem reduces to a standard (constrained) maximization problem for the geometric
factor~$g$, which in turns splits into two separate problems for the function
$g_1(\bu,\bv)\equiv\sum_{i=1}^r\log\sin\big(\tfrac{v_i-u_i}2\big)+\log f(\bu,\bv)$ and its
momentum space counterpart. For instance, when~$r=2$ we can express~$g_1(\bu,\bv)$ in terms of the
length~$L_1\equiv V_1-U_1$ of the first block and the interblock distance~$d\equiv U_2-V_1$ as
\begin{equation}\label{hthde}
  g_1(\bu,\bv)=\si(\th)+\si(2\pi\ga_x-\th)+\si(2\pi\ga_x+\de)+\si(\de)-\si(\th+\de)-\si(2\pi\ga_x-\th+\de)\equiv
  h(\th,\de)\,,
\end{equation}
where $\si(x)\equiv\log\sin(x/2)$, $\th=2\pi L_1/N\in(0,2\pi\ga_x)$,
$\de=2\pi d/N\in(0,2\pi(1-\ga_x))$. Moreover, from the symmetry of~$h$
under~$\th\mapsto2\pi\ga_x-\th$ and~$\de\mapsto2\pi(1-\ga_x)-\de$, it suffices to find the maximum
of this function in the rectangle~$(0,\pi\ga_x]\times(0,\pi(1-\ga_x)]$. An elementary calculation
shows that $h$ has a local maximum at $\th=\pi\ga_x$, $\de=\pi(1-\ga_x)$, i.e., at the symmetric
configuration, and that~$\nabla h$ has no other zeros on~$(0,\pi\ga_x]\times(0,\pi(1-\ga_x)]$.
This proves the conjecture in the case~$r=2$ (cf.~Fig.~\ref{fig.entcmp} (b)). For instance,
for~$r=s=2$ the maximum value of the entropy is easily found from the latter argument and
Eq.~\eqref{Sstruct} to be $4[b_\al\log(N\sin(\pi\ga_x)\sin(\pi\ga_p)/2\pi)+c_\al]$\,.
\begin{figure}[ht]%
  \includegraphics[height=5.5cm]{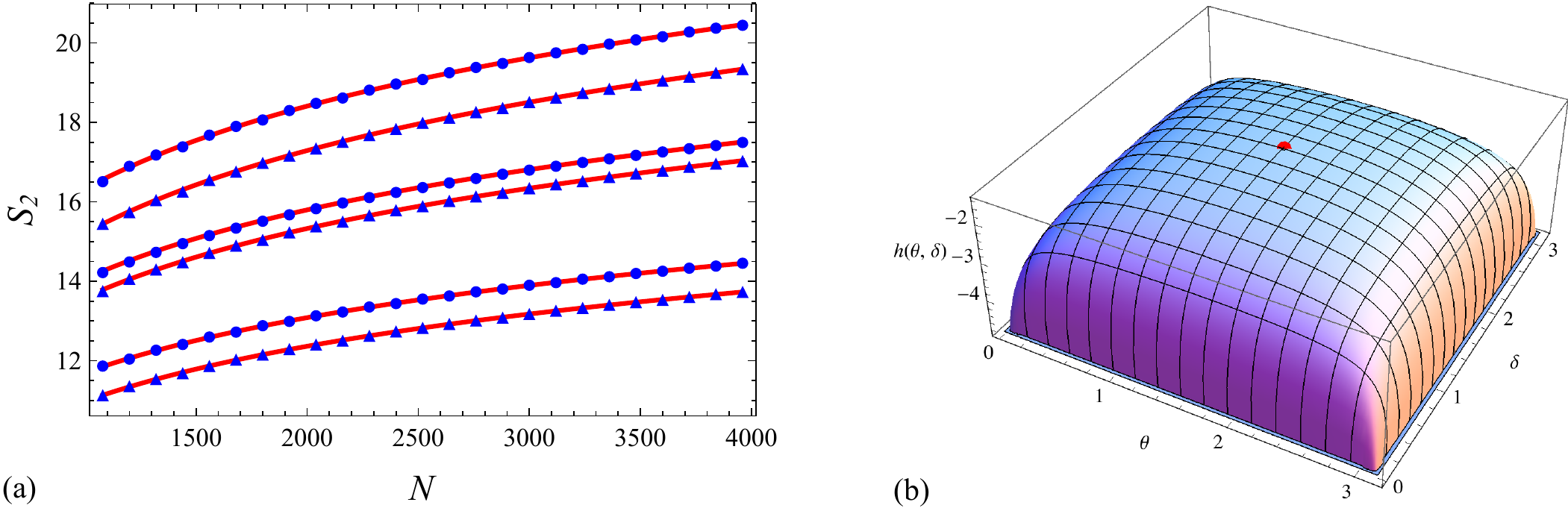}
  \caption{(a) R\'enyi entropy $S_2$ vs.\ its asymptotic approximation~\eqref{Salr1} (red line) in
    symmetric (blue points) and some non-symmetric (blue triangles) configurations
    with~$\ga_x=1/3$, $\ga_p=1/2$ and (bottom to top) $4+2$, $5+2$ and $4+3$ blocks. (b) 3D plot
    of the function~$h(\th,\de)$ in Eq.~\eqref{hthde} for $\ga_x=1/2$ (the red point corresponds
    to the symmetric configuration $(\th,\de)=(\pi/2,\pi/2)$).}
\label{fig.entcmp}
\end{figure}

\section*{Discussion}

In this work we have rigorously formulated a general duality principle which posits the invariance
of the R\'enyi entanglement entropy~$S(A;K)$ of a chain of free fermions under exchange of the
sets of excited momentum modes~$K$ and chain sites~$A$ of the subsystem under study, where both
$A$ and~$K$ are the union of an arbitrary (finite) number of blocks of consecutive sites or modes.
By means of this principle, we have derived an asymptotic formula for the R\'enyi entanglement
entropy when the set $K$ consists of a single block. From this formula and a natural assumption
concerning the additivity of the entropy when the blocks are far apart from each other in either
position or momentum space we have conjectured an asymptotic approximation for the entanglement
entropy in the general case when both sets~$A$ and~$K$ consist of an arbitrary number of blocks.
We have presented ample numerical evidence of the validity of this formula for different
multi-block configurations, and have analyzed its error comparing it with its counterpart for
the~$XX$ model discussed by Calabrese and Essler~\cite{CE10}. Our conjecture also yields an
asymptotic formula for the mutual information of a certain number of blocks in position (or
momentum) space valid for arbitrary multi-block configurations, which for $s=1$ and in the case of
an infinite chain is consistent with the general one for $1+1$ dimensional CFTs.

The previous results open up several natural research avenues. In the first place, it would be
desirable to find a rigorous proof of the fundamental asymptotic relation~\eqref{genconj}, which
leads to the explicit asymptotic formula~\eqref{conjfinal}. In particular, it would be of interest
to determine the range of models for which this relation holds. Another related problem is to
study analytically the precise behavior of the error term in the latter equation. Indeed, our
numerical results suggest that this error exhibits a qualitatively similar but considerably more
complex behavior than its analogue for an infinite chain with a single block in both position and
momentum spaces studied in Ref.~\cite{CE10}. Finally, an interesting question arising from the
discussion after Eq.~\eqref{Sstruct} is the analysis of the configurations \emph{minimizing} the
entropy with appropriate constraints, which could be naturally regarded as akin to
``semiclassical'' states.


\bibliography{cmprefs}

\section*{Acknowledgements}

This work was partially supported by Spain's MINECO under research grant no.~FIS2015-63966-P. PT
has been partly supported by the ICMAT Severo Ochoa project SEV-2015-0554 (MINECO, Spain). JAC
would also like to acknowledge the financial support of the Universidad Complutense de Madrid
through a 2015 predoctoral scholarship.

\section*{Author contributions statement}

J.A.C, F.F., A.G.-L. and P.T. contributed equally to this work. 

\section*{Additional information}
\textbf{Competing financial interests:} the authors declare no competing financial interests.

\end{document}